\documentclass[11pt]{article}
\usepackage[margin=1in]{geometry}

\usepackage{amsmath,amssymb,amsthm}
\usepackage{tikz}
\usetikzlibrary{arrows}
\usetikzlibrary{decorations}
\usetikzlibrary{shapes.misc}
\usepackage{float}
\usepackage{appendix}
\usepackage{url}
\usepackage{refcount}
\usepackage{xspace}
\usepackage{subfigure}
\usepackage{nicefrac}
\usepackage[inline]{enumitem}
\usepackage[mathscr]{euscript}

\usepackage{latexsym}
\usepackage{color}
\usepackage{graphicx}

\newcommand{\old}[1]{}

\renewcommand{\emph}[1]{\textit{#1}}
\newsavebox{\mycases}

 \allowdisplaybreaks

\newcounter{rot}

\def\occV{{O_v}}
\def\occP{{O_P}}

\newcommand{\ignore}[1]{}

\def\ii_(#1,#2){i_{#1}^{#2}}

\def\a{\alpha}
\def\b{\beta}
\def\d{\delta}
\def\D{\Delta}
\def\e{\varepsilon}

\def\Th{\Theta}
\def\l{\lambda}

\def\om{\omega}
\def\Om{\Omega}

\newcommand{\rdup}[1]{\left\lceil #1 \right\rceil}
\newcommand{\rdown}[1]{\mbox{$\left\lfloor #1 \right\rfloor$}}

\newcommand{\whp}{{\em w.h.p.}\xspace}

\newcommand{\ra}{\longrightarrow}
\newcommand{\rai}{\ra \infty}

\newcommand{\ooi}{(1+o(1))}

\newcommand{\brac}[1]{\left( #1 \right)}

\renewcommand{\Pr}{{\bf Pr}}
\newcommand\bfrac[2]{\left(\frac{#1}{#2}\right)}

\theoremstyle{plain}

\newtheorem{theorem}{Theorem}
\newtheorem*{theorem*}{}
\newtheorem{lemma}[theorem]{Lemma}

\theoremstyle{remark}

\newcommand{\nospace}[1]{}

\newcommand{\beq}[1]{\begin{equation}\label{#1}}
\def\eeq{\end{equation}}

\def\ep{\epsilon}

\newcommand{\E}{\mathbf{E}}

\let\epsilon=\varepsilon

\newcommand{\DD}{\mathbf D_{\text{Disp}}}
\newcommand{\TD}{\mathbf T_{\text{Disp}}}

\newcommand{\NumParticles}{M}  
\newcommand{\NP}{\NumParticles}

\begin{document}
\makeatletter
\title{
Dispersion processes\thanks{
This work was supported by EPSRC grant EP/M005038/1,
``Randomized algorithms for computer networks'',   and Becas CHILE.
}   }
\author{
Colin Cooper\thanks{Department of Informatics, King's College London, UK.
{\tt colin.cooper@kcl.ac.uk}}
\and Andrew McDowell\thanks{Department of Informatics, King's College London, UK.
{\tt andrew.mcdowell@kcl.ac.uk}}
\and Tomasz Radzik\thanks{Department of Informatics, King's College London, UK.
{\tt tomasz.radzik@kcl.ac.uk}}
\and Nicol\'as Rivera\thanks{Department of Informatics, King's College London, UK.
{\tt nicolas.rivera@kcl.ac.uk}}
\and Takeharu Shiraga\thanks{
Department of Information and System Engineering, Chuo University, Japan.
{\tt shiraga@ise.chuo-u.ac.jp}}
}

\maketitle \makeatother
\begin{abstract}
We study a synchronous dispersion process in which $\NP$ particles are initially placed at
a distinguished {\em origin vertex\/} of a graph $G$.
At each time step, at each vertex $v$ occupied by more than one particle at the beginning of this step,
each of these particles moves to a neighbour of $v$ chosen independently and uniformly at random.
The dispersion process ends
at the first step when
each vertex has at most one particle.

For the complete graph $K_n$ and star graph $S_n$, we show that for any constant $\d>1$, with high probability,
if $\NP \le n/2(1-\d)$, then the process finishes in $O(\log n)$ steps,
whereas if $\NP \ge n/2(1+\d)$, then the process needs $e^{\Omega(n)}$ steps to complete (if ever). We also show that an analogous lazy variant of the process exhibits the same behaviour but for higher thresholds, allowing faster dispersion of more particles.
For paths, trees, grids, hypercubes and Cayley graphs of large enough sizes (in terms of $\NP$)
we give  bounds on the time to finish and
the maximum distance traveled from the origin as a function of the number of particles $\NP$.


\end{abstract}

\newpage

\section{Introduction}
A dispersion process can be  described as follows.
Initially  a group of identical particles are located at a single vertex of a graph. The particles move apart in a distributed fashion until
 no more than one particle occupies any vertex. When this occurs we say the particles are dispersed.

We require the behaviour of the particles during dispersion to be  identical,  their movements random, and that no communication, prioritization or other symmetry breaking occurs.
The process we consider, hereafter called Dispersion, works as follows. The process is synchronous and proceeds in discrete steps.
Whenever two or more particles occupy the same vertex at some step, they move independently to a random neighbour. If only a single particle occupies a vertex, it  stays there until another particle arrives.
Thus particles move as a reflex action when two or more particles occupy the same position.
If we reach a situation where each particle is on a different vertex,  there can be no further movement and the particles have dispersed.

For each step at which it moves, each particle makes an independent random walk.
However, the steps at which a particle moves are completely correlated with the arrival of other particles, and so the particles make random walks which stop and start.
Up to the time a particle finally stops, at every step it moved someone else moved with it.

Dispersion is an abstraction of many situations. The simplest ones are from physics. For example when a group of similarly charged particles are held at a single point, and move apart by natural repulsion. Another, concerns dispersion of hard spherical particles (atoms)  which do not allow spatial overlap.
An example from biology is  dispersion of progeny; a clutch of eggs  hatch, and the hatchlings move away, each to establish an exclusive territory.

The dispersion process differs from the type of methods considered previously
for dispersing robots or sensors,
in that we do not explicitly require the particles to disperse uniformly throughout the network,
but merely to move away from one another and establish a personal space.
This means that the degree of self-organization is less than required for swarm systems.
Random  dispersion of swarms  is considered by  \cite{Beal, Beal1}.
The particles use L\'evy Flights to move a random biassed distance $d$ with probability proportional to $1/d$
(within some large finite range).

We suggest that dispersion could be used as a primitive form of load balancing in the absence of symmetry breaking. This assumes that any vertex is willing to process at most one job, but no vertex is prepared to process two or more jobs.

Dispersion is in many ways  a natural  analogue of
Internal Diffusion Limited Aggregation (IDLA).
In IDLA  particles start from the origin vertex one at a time. The next particle does not start until the previous particle stops moving.
Once introduced, the current particle moves randomly until it reaches an unoccupied vertex. It then occupies the vertex permanently and does not move any further. Subsequent particles which arrive at an occupied vertex
continue to walk randomly until they arrive at a vacant site. The process stops when
the last particle settles at a vertex. The IDLA  process was introduced by Diaconis and Fulton. Their paper, \cite{DF}, gives the limiting shape made by the particles on  the integer line. Lawler, Bramson and Griffeath \cite{LBG} subsequently generalized the analysis to $d$-dimensional grids. For two dimensional grids they proved the limiting shape is a disk.  There is also interest in the shapes made by the corresponding  rotor-router analogue of IDLA. For two dimensional grids, Levine and Peres \cite{LP} proved the limiting rotor-router shape is spherical.

This suggests a synchronous version of IDLA, in which if a single particle
occupies a vertex at any step it halts permanently, whereas if two or more particles
arrive simultaneously, or a particle arrives at an already occupied vertex, the new arrival moves independently to a random neighbour. However, because the particles are allowed to behave in an asymmetric fashion (stopping permanently on single occupancy) this model seems less satisfactory than Dispersion.  
In synchronous IDLA, particles move at every step until they stop permanently, whereas in Dispersion particles stop and start, and a temporarily stopped particle can never know that it will move again.
This makes it difficult to relate the walk steps of a particle to the steps of the Dispersion process, and presents an additional obstacle to analysis.

We analyse the synchronous Dispersion process in which $\NP$ particles are initially placed at a single vertex of a graph $G$, which we  call the  origin vertex.
If two or more particles occupy the same vertex
at the end of  step $t$,  then all particles at that position move independently to a random neighbour at step $t+1$. Thus it can be that (by chance) the particles move to the same place and have to move again at step $t+2$, and so on.
 The process ends once the particles have all stopped moving.
This occurs when all vertices are occupied by either one or no particles. Trivially, for the process to end $G$ must have at least as many vertices as there are particles.

We are interested in properties of the process such as the distance the
 particles travel from the origin and the time taken to disperse.
The dispersion time $\TD$ is the number of synchronous time steps taken to
disperse the particles. The dispersion distance $\DD$ is the maximum distance of any particle from the origin at dispersion.

We analyse the performance of dispersion on a number of different graphs,
including the complete graph, the star graph and sufficiently large paths, grids, hypercubes, Cayley graphs and regular trees.
The complete graph $K_n$ exhibits a threshold in dispersion time from $O(\log n)$
when the number of particles $\NP$ is at most $(1-\d)n/2$, to $e^{\Om(n)}$
when the number of particles $\NP$ is at least $(1+\d)n/2$, where $\d >0$ is an arbitrarily small constant.
%
The following theorem is proven in Section~\ref{SKn}.

\begin{theorem} \label{Kn}
For the complete graph $K_n$, and the star $S_n$ the following hold for any constant $\d>0$.

(i) If the number of particles $\NP$ satisfies $\NP/n \le (1/2)(1-\d)$, then
with probability $1 - O(1/n)$, the dispersion process
terminates in $\TD= O(\log n)$ steps.

(ii) If the number of particles $\NP$ satisfies $(1/2)(1+\d) \le \NP/n < 1$,
then there is a constant $c = c(\d) > 0$ such that the probability that
$\TD \le  e^{cn}$ is less than $e^{-cn}$.
\end{theorem}

We also consider a variant process which we call Lazy Dispersion. For some $0<p\leq 1$, a particle which occupies a vertex containing any other particles, instead moves with probability $p$ and stays at its current vertex with probability $1-p$, while particles which occupy a vertex alone, as before, do not move. This model represents a sliding scale, with the behaviour of the process becoming closer to that of IDLA as $p$ tends to $0$.  We prove the following analogous results to the above, demonstrating that a smaller $p$ allows logarithmic dispersion up to a higher threshold.
More precisely, the threshold for $\NP/n$ which separates fast and slow dispersion generalises from $1/2$ to $1 - p/2$.

\begin{theorem} \label{Thm:LazyKn}
Given $0<p\le 1$, which may depend on $n$,
the (lazy) dispersion of $\NP=(1-\delta)n$ particles on the complete graph $K_n$
behaves in the following way.

(i) If $\delta = \frac{p}{2}+\alpha$ for some $\alpha >0$ (which may depend on $n$), then
with probability $1 - O(1/n)$, the dispersion process
terminates in $\TD= O((p\a)^{-1}\log n)$ steps.

(ii) If $\delta = \frac{p}{2}-\alpha$ for some $\alpha >0$ (which may depend on $n$),
then there exists a constant $c>0$, such that the probability
that $\TD \le  e^{c n p^2\a^3}$ is less than $e^{- c n p^2\a^3}$.
\end{theorem}

We give a general result for random walks on grids, the hypercube and indeed any other symmetric Cayley graph of an Abelian group.
In this paper the term 'Cayley graph' refers to this type of Cayley graphs.
In such a graph, the simple random walk transition at any vertex is determined by sampling uniformly from
a symmetric generator set $S$  which defines the graph (symmetric means that if $g \in S$, then also  $-g \in S$).
Transitions at vertex $u$ are made to
$v=u+g$
(we use ``$+$'' to denote the group operation),
where $g\in S$ is the group element which labels edge $(u,v)$.
The edge $(v,u)$ from $v$ to $u=v + (-g)$
is also present.
For example, on the line the transitions are defined by $S = \{-1,+1\}$,
so denoting by $X_t$ the position of the random walk at step $t$,
we have $X_{t+1}=X_t \pm 1$, equiprobably.
For the two dimensional (infinite) grid, the transitions are defined by $S = \{ (1,0), (0, 1), (-1,0),(0, -1)\}$.

\begin{theorem} \label{Caley}
Let $\om = \om(\NP) \rai$.

(i) Let $G$ be a $d$-dimensional (infinite) grid ($d \ge 1$) or other infinite Cayley graph,
and let $t$ be such that $t \ge \om \NP^2 R(2t)$,
where $R(2t)$ is the expected number of returns to the origin in $2t$ steps by a simple random walk on $G$.
Then with probability at least $1 - 1/\om$,
a system of $\NP$ particles disperses on $G$ in $t$ process steps.

(ii) Let $G = (V,E)$ be the hypercube on $n=2^d$ vertices, or other $n$-vertex finite Caley graph, let $n' = n/2$, if
$G$ is bipartite, and $n' = n$, if $G$ is non-bipartite, let $P$ be the transition matrix of the random walk on $G$,
and finally, let
$T$ be a step of the random walk
such that for all even $s \ge T$ and all $u \in V$, $|P^s(u,u)-1/n'| \le 1/2n'$.
Then with probability at least $1 - 1/\om$,
a system of $\NP=o(\sqrt{n/\om})$ particles disperses on $G$ in $t=O(T \NP^2)$ process steps.
\end{theorem}

\begin{sloppy}

The values of $R(2t)$ for the line, 2-dimensional grid, and grids of dimension at least 3   are $\Th(\sqrt{t})$,
$\Th(\log t)$ and $\Th(1)$ respectively.
This gives values of $t=O(\om^2 \NP^4)$ for the line, $t=O(\om \NP^2 \log M)$ for the 2-dimensional grid
and $t=O(\om \NP^2)$ for grids of  dimension at least 3.
For  the hypercube, $T=O(\log^2 n)$, and for the $n$-cycle $T=O(n^2 \log n)$.
Thus, provided $\NP=o(\sqrt{n})$,
the dispersion time for the hypercube is $O(\NP^2 \log^2 n)$ and for the $n$-cycle is $O(\NP^2n^2\log n)$, with probability
at least $1 - o(1)$.
The proof of Theorem~\ref{Caley} is in Section~\ref{Grid}, part $(i)$, and in Section~\ref{Sec:Hypercube}, part $(ii)$.

\end{sloppy}

We next give more precise results for sufficiently large $k$-regular trees. The case $k=2$, the path graph,
differs from the case $k \ge 3$ and is stated separately.
The proofs are given in Sections  \ref{SkReg} and \ref{SPath}. To remove a factor of $\NP$ or so from the above results takes some work.

\begin{theorem}\label{Path}
For a sufficiently long  path, and $\NP$ particles initially placed at the central vertex of the path, the following holds w.h.p. for any  $\varepsilon>0$.
When the dispersion process terminates, the maximum distance $\DD$ any particle is from the origin is bounded by
\begin{equation}\label{Eq:path}
\rdown{\NP /2 }\le \DD \le 4(1+\varepsilon)\NP\log \NP,
\end{equation}
and $\TD =O( \NP^3 \log \NP)$.
\end{theorem}

The lower bound on the dispersion distance $\DD$
in Theorem \ref{Path} comes from the following simple observation, which applies to any graph.
Let $d(\NP)$ be the minimum graph distance about the origin $v$ such that the subgraph $S(v)$ induced by vertices of distance at most $d(\NP)$ from $v$  contains at least $\NP$ vertices, then $\DD \ge d(\NP)$. Thus for a path graph $\DD\ge \NP/2$, and for $k$-regular trees $\DD \ge \log_{k-1} \NP$.
The proof of the upper bound in~\eqref{Eq:path} implies that during the dispersion process on an infinite path \whp\
no particle is ever further away from the origin than at distance $6\NP\log \NP$.
That is, ``sufficiently long path'' in the statement of the theorem means a path of length at least $12\NP\log \NP$,
and the same bounds on $\DD$ and $\TD$ apply to cycles of at least this length.

\begin{theorem}\label{kReg}
Let $k \ge 3$.
There exist constants $0 < \beta_k < \alpha_k$, with $\alpha_k\rightarrow 0$ as $k\rightarrow \infty$,
such that for a sufficiently large  $k$-regular tree,  and $\NP$ particles initially placed at the central vertex of the tree,
the following
holds for any  constant $\varepsilon>0$ with probability $1 - O(\NP^{-\ep})$. When the dispersion process terminates, the maximum distance $\DD$ any particle
is from the origin is bounded by
\begin{equation}\label{Eq:tree}
\left(2-\alpha_k -\varepsilon \right)\log_{k-1} \NP \le \DD \le\left(2 -\beta_k +2\varepsilon\right)\log_{k-1} \NP,
\end{equation}
and $\TD =O(\NP \log_{k-1} \NP)$.
\end{theorem}

The proof of Theorem~\ref{kReg} uses infinite trees and shows not only that \whp\ no particle will end up at the termination of the
Dispersion process
further away from the origin than as given in~\eqref{Eq:tree}, but also that
\whp\ no particle will be at any time {\em during\/} the process further away from the origin than this bound.
Thus a ``sufficiently large $k$-regular tree'' in the wording of Theorem~\ref{kReg} means a complete $k$-regular tree
of depth at least as the upper bound on $\DD$ in~\eqref{Eq:tree}
(all leaves are at the same distance from the central root vertex and each internal node, including the root, has
degree~$k$).


To distinguish between the dispersion process and the random walks made by the particles we call the steps $t=0,1,2,...$ of the dispersion process, {\em time steps\/} or {\em process steps},
and the steps of the walks, {\em walk steps}. The time steps $t$ go on forever, but after dispersion the particle
locations do not change.

Let $N_i(t)$ be the number of walk steps taken by particle $i$ at or before time step $t$. At each time step  during dispersion at least two particles
move.
Thus $\sum_{i=1}^{\NP} N_i(t) \ge 2t$, and for some particle $i$, $N_i(t)\ge 2t/\NP$.
If \whp\ no particle makes more than $x$ walk steps during the dispersion process, then
\whp\ $\TD \le Mx/2$.
For example, in the proof of Theorem~\ref{kReg}, the bound on $\TD$ follows from a
$O(\log_{k-1} \NP)$ bound on the number of walk steps any particle can make before the dispersion terminates.
In the proof of Theorem~\ref{Path}, the core argument is that
each particle makes $O(\NP^2\log \NP)$ walk steps. This implies
an upper bound $O(\NP^3\log \NP)$ on $\TD$ and an upper bound $O(\NP\log \NP)$ on $\DD$.


It is a condition of the dispersion process that the particles make
independent random walks whenever they move on the underlying graph $G$. To remove any suspicion of correlation between the walks we adopt the following device, and
predetermine the movements the particles will take when they move. For each particle, independently predetermine an infinite random walk on $G$. Whenever the particle is required to move in the dispersion process, it reads the next movement from its own random walk and follows it. In this way the walk is independent of the dispersion process and the movement of any other particle. However the number of steps taken by the walk at a given step of the dispersion process and the step of the walk at which the particle will stop forever is entirely determined by the underlying dispersion process.

Theorem~\ref{Kn} says that the Dispersion process on the complete graph $K_n$ has the threshold at $\NP \sim (1/2)n$.
It seems reasonable to ask if the existence of such a threshold is a general phenomena, and if so,
to define the  dispersion number of a finite graph as the (limiting) maximum proportion of
particles which can be dispersed on  graphs of this type in an expected number of steps polynomial in the number of vertices.
For $K_n$ the dispersion number is $1/2$, whereas experimentally, the dispersion number of the cycle $C_n$ is at least $0.89$.
As mentioned earlier, the proof of Theorem~\ref{Path} implies that $\NP = (c/\log n) n$ particles disperse on the $n$-vertex cycle
in polynomial time.
We leave as an open question whether cycles or other graphs (significantly different than complete graphs and stars)
have constant dispersion numbers.

\paragraph{Proof methodology.}

In this paper we utilise a number of different methods and techniques to analyse the
behaviour of the dispersion process on different graph structures.
Our methods in bounding the maximum distance a particle can travel largely fall into two categories.

The first, which we utilise on the line, grid and Cayley graphs involves bounding the number of meetings between particles, which allows us to find a bound on the total running time of the process. We utilise that these graphs have high levels of symmetry allowing us to treat the number of meetings of two particles as the number of returns to the origin of a combined walk, reversing one of the particles movements. We know that if the process has not ended, then at least one meeting of two particles occurs for every time step that takes place but if the number of meetings a particle encounters grows much more slowly than the number of steps it takes, this will lead us to a contradiction. For the path we are able to give reasonably tight bounds on the number of meetings of two particles, giving bounds close to the right order of magnitude. For other graphs, determining better bounds on the number of meetings of two particles would immediately allow for better bounds on the dispersion time.

Our second method, which we utilise on the $k$-regular tree, is to take advantage of the branching structure and the fact that there are many distinct vertices a large distance from the origin. In particular, we use that at some point for a particle to reach a distance $d$ and continue moving, another particle must also visit the same vertex at $d$. More strongly it is possible to say that there must exist two particles, who when they first reach any vertex at distance $d$, do so at the same vertex. Since in the tree it is equally likely that a particle ends up at any particular vertex at this depth first, and the number of such vertices is much higher than the number of particles, the probability of this occurring is small.

The first method allows us to bound the total running time of the process, which we may be able to use in turn to bound the distance a particle can move. Conversely, the second gives direct bounds on the distance a particle can move, from which we may also be able to deduce bounds on the running time.

The hypercube is an interesting example because it satisfies both of these properties and is a good example of how these two methods can be applied.
In Section~\ref{Sec:Hypercube}, we derive bounds $O(\NP^2 \log^2 n)$ and $O(\NP\log^3 n)$
using these two methods, respectively.

\section{Dispersion on the complete graph } \label{SKn}
A proof of Theorem \ref{Kn} for the complete graph $K_n$ is given in this section.
To keep the proof tidy, we first analyse the case of $K_n$ with loops. Details for $K_n$ without loops are given afterwards.

We say a particle is happy at step $t$
if it is the only particle at its vertex, and unhappy otherwise. A particle which is happy stays put and only moves if another (currently unhappy) particle moves to the vertex it occupies.

The behavior of dispersion on the star graph $S_n$ is almost identical to that on $K_n$ (without loops). It follows from the observation that, at alternate steps any unhappy particles congregate at the central vertex, and then jump to a random leaf.

{\bf Proof of Theorem \ref{Kn}, case of $K_n$ with loops.}
Let $\NP$ be the total number of particles.
A particle is happy (at a given step) if it is the only particle at its current vertex, otherwise it is unhappy. Let
 $H(t)$ be
the number of happy particles at step $t$, and $U(t)$ the  unhappy ones.
Thus $H(t)+U(t)=\NP$. The process ends when $H(t)=\NP$.

In what follows, we bound the value of $H(t+1)$ given the values of
$H(t)$ (and thus $U(t)$) at step $t$. 
At each time step $t$ any  unhappy particle moves to a random vertex $v \in [n]$. The particles which are happy do not move.
Suppose $U(t)>0$. At the next step there are $X=X(t)$ previously happy particles which became unhappy because (unhappy) particles landed on top of them. Also $Y=Y(t)$ previously unhappy particles became happy by being the only particle to move to one of the $n-H(t)$ unoccupied vertices. This gives $H(t+1)=H(t)-X+Y$.
We obtain $\E(X(t) \mid H(t)), \E(Y(t) \mid H(t))$ and hence $\E(H(t+1) \mid H(t))$. 
To simplify notation we do not explicitly state the conditioning on $H(t)$,
and we abbreviate $H(t)$ and $H(t+1)$ to $H$ and $H'$, respectively.

The properties of $X, Y$ are as follows. 
If we randomly allocate $U$ balls to $n$ boxes, of which $H$ are non-empty, 
the number $X$ of the  non-empty boxes $H$ receiving at least one ball (resp.
the number of empty boxes  $Y$ receiving exactly one ball) have expected values
\begin{eqnarray}
\label{EX}
\E X &=& H \brac{1-\brac{1-\frac1n}^U},\\
\label{EY}
\E Y &=& U\bfrac{n-H}{n} \brac{1- \frac{1}{n}}^{U-1}.
\end{eqnarray}

\begin{sloppy}

The concentration of 
$H' = H - X + Y$ follows from considering 
the Doob martingale 
$B_i = \E_{Z_{i+1},\ldots,Z_{U}} [H'| Z_1, Z_2, \ldots , Z_{i}]$, where $Z_i$ is the box (vertex) chosen by the $i$-th unhappy ball.
Thus $B_0 = \E[H']$, $B_{U} = H'$, $\E B_i = B_{i-1}$, and
$|B_{i} - B_{i-1}| \le 2$ because a difference in choice of bin by ball $i$ (with all other choices remaining the same) 
can only alter the value of $H'$ by at most 2. 
The Azuma-Hoeffding inequality implies
\begin{equation}\label{martingale}
\Pr(|H'-\E H'| \ge \l) \le 2 \exp \brac{-\frac{\l^2}{8U}}.
\end{equation}
\end{sloppy}


Let $\D H=H(t+1)-H(t)=Y-X$.
\begin{flalign}
\E \D H &= \frac{U}{n}(n-H) \brac{1- \frac1n}^{U-1} -H\brac{1-\brac{1-\frac1n}^U} \label{EDH1}\\
&= \brac{1-\frac1n}^U\brac{U+H- \frac{U(H-1)}{n-1}}-H, \label{EDH2}
\end{flalign}

{\bf Case $\NP\le(1-\d)n/2$.}\\
 For $U \ge 0$, $(1-1/n)^U \ge (1-U/n)$.
Substituting this into \eqref{EDH2}, and using $U+H=\NP$ gives
\begin{flalign*}
\E \D H &\ge \brac{ 1-\frac{U}{n}}\brac{U+H-\frac{UH}{n}}-H\\
&= U \brac{ 1-\frac{\NP}{n}-\frac{\NP-U}{n}\brac{1-\frac{U}{n}}}.
\end{flalign*}
Thus
\[
\E (U(t+1) \mid U(t)) \le \frac{U}{n}\brac{\NP+(\NP-U) \brac{1-\frac{U}{n}}} \le U(t) \frac{2M}{n}.
\]
Now use $\NP \le (n/2)(1-\d)$ and iterate to get
\[
\E U(t) \le U(0)\brac{1-{\d}}^t \le U(0) \exp\brac{-{t \d}}.
\]
Choosing
\begin{equation}\label{Less}
t=(2/\d) \log n
\end{equation}
gives $\E U(t)=O(1/n)$ and thus $U(t)=0$, and hence $\TD \le t$, with probability $1- O(1/n)$.

{\bf Case $\NP\ge (1+\d)n/2$.}\\
Let $H(t)=(n/2)(1+\e)$,
where  $-1\le \e \le \d/2$.
We prove below that
\begin{equation}\label{ED}
\E H(t+1) \le \ooi \frac{n}{2} \brac{
1+\frac{\d}{2} \brac{1-\frac{3\d}{8}}}.
\end{equation}
By  the concentration of $H(t+1)$
(see discussion below \eqref{EY}),
for $\e \le \d/2$
\begin{equation}\label{More}
\Pr\brac{H(t+1) \ge \frac{n}{2} \brac{
1+\frac{\d}{2}}}
 \le   e^{-2cn},
\end{equation}
for some constant $c = c(\d) > 0$.
To disperse the particles requires $H$ to equal $\NP=(n/2)(1+\d)$.
The Inequality~\eqref{More} implies, however, that
$H$ remains below $(n/2)(1+\d/2)$ for $e^{cn}$ steps with
probability at least $(1 - e^{-2cn})^{e^{cn}} \ge 1 - {e^{cn}}$.

{\bf Proof of  equation \eqref{ED}.}
 From \eqref{EDH2}, with $U=(n/2)(\d-\e)$,  we have that
\[
\E H(t+1)\le  \ooi e^{-\frac{U}{n}}\brac{\NP-\frac{U(\NP-U)}{n}}= \ooi \frac{n}{2} A(\e, \d),
\]
say, where
\begin{flalign*}
A(\e,\d) & \le e^{-\frac12(\d-\e)}\brac{(1+\d)-\frac12(\d-\e)(1+\e)}\\
&= e^{\e/2-\d/2}\brac{1+\frac{\d}{2}+\frac{\e}{2}(1-\d)+\frac{\e^2}{2}}.
\end{flalign*}
Thus $A(\e,\d)$ is monotone increasing in $\e$ and for $\e \le \d/2$,
\[
A(\e,\d) \le A(\d/2,\d)= e^{-\d/4}\brac{1+\frac{3\d}{4}-\frac{\d^2}{8}}.
\]
For $0 \le x \le 1$, $e^{-x} \le 1-x+x^2/2$, so that
\begin{flalign*}
A(\d/2,\d) & \le \brac{1-\frac{\d}{4}+\frac{\d^2}{32}}\brac{1+\frac{3\d}{4}-\frac{\d^2}{8}}\\
&= 1+\frac{\d}{2}-\frac{\d^2}{32}\brac{9-\frac{7\d}{4}+\frac{\d^2}{8}}.
\end{flalign*}
However $9-\frac{7\d}{4}+\frac{\d^2}{8}$ is monotone decreasing in $\d$ for $0 \le \d \le 1$ and
\[
A(\d/2,\d) \le 1+\frac{\d}{2} \brac{1-\frac{6\d}{16}}.
\]

{\bf Case of  $K_n$ without loops.} For $\E X$ in \eqref{EX},  the value of $(1-1/n)$  becomes $(1-1/(n-1))$.
The effect on $Y$ is to slightly increase the value of $\E Y$ in \eqref{EY} as follows.
Let $V(H)$ be the happy vertices, and $u(v)$ the number of unhappy particles at $v \in V$. The upper tail of $u(v)$ is stochastically dominated by $Z \sim Bin(\NP,1/(n-1))$. Using a Chernoff bound that $\Pr(Z \ge \a \mu)\le (e/\a)^{\a\mu}$
\[
\Pr(Z \ge n/\log n) \le \bfrac{Me\log n}{n(n-1)}^{n/\log n}=e^{-\Th( n)}.
\]
Thus (w.h.p.)
\begin{flalign*}
\E Y &=\sum_{i \in U} \sum_{v \in [n]-V(H)-\{v_i\}} \frac{1}{n-1}\brac{1-\frac{1}{n-1}}^{U-u(v)-1}\\
&=(1+O(1/\log n))\;U\brac{1-\frac{H}{n}}\brac{1-\frac1n}^U,
\end{flalign*}
where $v_i$ is the vertex currently occupied by unhappy particle $i$.
The rest of the proof is the same.

\subsection{Lazy Dispersion on $K_n$}\label{Section:LazyDispersion}
We have shown that  on the complete graph, the dispersion process disperses the particles in logarithmic time (w.h.p.) if the number of particles is less than half the number of vertices. If the number of particles is more than half the number of vertices, there is  a double exponential leap, as it now requires exponential time (w.h.p.) to disperse the particles.
We next show that, perhaps counter intuitively, slowing down the particles can allow the process to disperse more quickly. More precisely,
we show that if instead of all unhappy particles moving, each unhappy particle moves with some probability
$0 < p\le 1$, then for suitable choices of $p$, we can disperse many more particles in logarithmic time.

To have some intuition as to why slowing particles down may speed up the process, consider that for small enough $p$, we can assume at any time step that at most one particle moves with high probability. In this range, we have a process in which particles that are happy stay still and at most one unhappy particle moves. This ensures that any vertex that is occupied will never become unoccupied. This is identical, other than the order that the particles move and that at some time steps, nothing changes, to the IDLA process, which we know completes in polynomial time, even for $\NP=n$, and the process has only been slowed by a factor of $p$.

As before, we let $\NP$ be the total number of particles.
A particle is happy (at a given step) if it is the only particle at its current vertex, otherwise it is unhappy. Let
 $H(t)$ be
the number of happy particles at step $t$, and $U(t)$ the  unhappy ones.
Thus $H(t)+U(t)=\NP$. The process ends when $H(t)=\NP$.
At each time step $t$ any  unhappy particle moves to a random vertex $v \in [n]$ with probability $p$
and stays still with probability $(1-p)$ independently of any other particle. The happy particles  do not move.

%
%
%

\subsection*{Proof of Theorem \ref{Thm:LazyKn}(i).}

Recall that we have $\NP=(1-\delta)n$ particles moving on $K_n$ (with loops) where $\delta = \frac{p}{2}+\alpha$.

The proof of Theorem \ref{Thm:LazyKn} part (i) uses different method from the the proof of Theorem \ref{Kn}.  The distribution of unhappy particles at a given vertex affects the probability of a single unhappy particle becoming happy by remaining at this vertex,  while the other particles at the location all leave.

For a vertex $v$, let $\occV$ be the number of particles at $v$. Let $\mathcal{V}$ be the set of vertices with unhappy particles, and $E$ be the number of vertices such that $\occV=0$. 
Let $R=H+|\mathcal{V}|$ be the range of occupied vertices, i.e. the number of vertices with non-zero occupancy. Similarly to the above, at each time step, $R$ may change either positively due to a vertex that is unoccupied receiving at least one particle, or negatively due to a (necessarily unhappy) vertex losing all of its current particles. Let $R_+$ and $R_-$ represent these values respectively.

Using $(1-x)^k \le 1 - kx + k^2x^2/2$, if $k$ is a positive integer and $kx \le 1$,
we have,

\begin{eqnarray}
\label{eqn:ER+}
\E R_+ &=& E \brac{1-\brac{1-\frac{p}{n}}^U}
\; \geq \; \frac{EUp}{n}\left(1-\frac{Up}{2n}\right). \\
%
\nonumber
\E R_- &=& \sum_{v\in \mathcal{V}} p^\occV \left(1-\frac{1}{n}\right)^\occV \left(1-\frac{p}{n}\right)^{U-\occV}
\;\;  = \;\;  \left(1-\frac{p}{n}\right)^{U}\sum_{v\in \mathcal{V}} \left( \frac{p(n-1)}{n-p}\right)^\occV \\\nonumber
&\leq&  \left(1-\frac{Up}{n}+\frac{U^2p^2}{2n^2}\right)\sum_{v\in \mathcal{V}} p^\occV
\; \; \leq \; \;  \left(1-\frac{Up}{n}+\frac{U^2p^2}{2n^2}\right)\frac{U}{2}p^2\\
&=&\frac{Up^2}{2}-\frac{U^2 p^3}{2n}+\frac{U^3p^4}{4n^2}
\;\; = \;\;
\frac{Up^2}{2}\left(1 - \frac{U p}{n}\left(1- \frac{U p}{2n}\right)\right) 
\label{eqn:ER-}
\;\; \le \;\; \frac{Up^2}{2}\left(1 - \frac{U p}{2n}\right).
\end{eqnarray}

We therefore have that the expected change in $R$ satisfies,
\begin{eqnarray}\nonumber
\E \D R &=& \E [R_+  -  R_-]
\\\nonumber
&\geq &
\frac{EUp}{n}\left(1-\frac{Up}{2n}\right) - \frac{Up^2}{2}\left(1 - \frac{U p}{2n}\right)
\\\nonumber
&=& Up \left(1-\frac{Up}{2n}\right) \left(\frac{E}{n}- \frac{p}{2}\right).
%
\end{eqnarray}

We consider the case $\d = p/2 + \a $, so $E/n\geq \delta = p/2 + \a$.
We also have $(Up)/(2n) \le 1/2$, so
\begin{equation}\label{eqn:EDEltaR}
\E \D R \geq Up\a/2.
\end{equation}


We define
$D(t) \equiv \NP-R(t)$.
The change of $D$ in each time step is equal to the negative of the change in $R$ and we have $D=0$ if and only if $R=\NP$ and dispersion has occurred.
By \eqref{eqn:EDEltaR} and using $D(t)\leq U(t)$, we have that

\[
\E (D(t+1) | D(t))\leq D(t)-U(t)p\a/2 \; \le \; D(t) \left(1-p\a/2\right).
\]

Iterating our argument we have,

\[
\E D(t) \le D(0)\left(1-p\a/2\right)^t \le D(0) \exp\brac{-{t p\a/2}}.
\]
Clearly $D(0)<n$ and so choosing
\begin{equation}\label{Less}
t=\frac{ 4\log n}{p\a} =O((p\a)^{-1}\log n),
\end{equation}
gives $\E D(t)=O(1/n)$.
Hence $\TD \le t$, with probability $1- O(1/n)$ as required.

\subsection*{Proof of Theorem \ref{Thm:LazyKn} (ii).}

Recall that we have $\NP=(1-\delta)n$ particles moving on $K_n$ (with loops)
where $\delta = \frac{p}{2}-\alpha$.

The proof of Theorem \ref{Thm:LazyKn} part (ii) is similar to  that of Theorem \ref{Kn}. We measure the expected change in happy and unhappy particles. However the calculations are more involved.


Suppose $U(t)>0$. At the next step there are $X$ previously happy particles which became unhappy because unhappy particles moved and  landed on top of them. Also $Y$ previously unhappy particles became happy by either being the only particle to move to an unoccupied vertex (either one that was already unoccupied or one which contained several unhappy particles which all moved), or by being the only unhappy particle which did not move from its current location
or moved to the same location (chose to move but followed the loop edge),
while all other particles at that vertex left.
This gives $H(t+1)=H(t)-X+Y$. 


For each unhappy particle $P$,
let $\occP$ be the number of other particles at the same vertex.
For a vertex $v$, let $\occV$ be the total number of particles at $v$. Let $\mathcal{U}$ be the set of unhappy particles, $\mathcal{V}$ be the set of vertices with unhappy particles, and $E$ be the number of vertices such that $\occV=0$.
We therefore have $n=H+E+|\mathcal{V}|$.
The discussion above implies the following formulas for the expected values of $X$ and $Y$.
Inequality~\eqref{kleoa3s} would become equality, if we excluded from the inner sum the vertex occupied by the particle $P$.
\begin{eqnarray}
\label{EXLAZY}
\E X &=& H \brac{1-\brac{1-\frac{p}{n}}^U}=H-\frac{H(n-p)}{n}\left(1-\frac{p}{n} \right)^{U-1}.
\end{eqnarray}
\begin{eqnarray}
\label{kleoa3s}
\E Y & \le & \sum_{P\in \mathcal{U}}\left( E\bfrac{p}{n} \brac{1- \frac{p}{n}}^{U-1}+
 \left(1-p+\frac{p}{n}\right)p^{\occP} \left( 1-\frac{1}{n}\right)^{\occP}\left( 1-\frac{p}{n}\right)^{U-\occP-1}\right) \\\nonumber
& & +  \sum_{P\in \mathcal{U}}
\sum_{v\in \mathcal{V}}\left(p^{\occV} \left( 1-\frac{1}{n}\right)^{\occV}
\bfrac{p}{n}\left( 1-\frac{p}{n}\right)^{U-\occV-1} \right)
\\\nonumber
&=&  \frac{ UE p}{n} \brac{1- \frac{p}{n}}^{U-1}+\frac{Up}{n}\sum_{v\in \mathcal{V}}\left(p^{\occV} \left( 1-\frac{1}{n}\right)^{\occV}\left( 1-\frac{p}{n}\right)^{U-\occV-1} \right)\\\nonumber
& & + \sum_{v\in \mathcal{V}}\occV\left(\left(1-p+\frac{p}{n}\right)p^{\occV-1} \left( 1-\frac{1}{n}\right)^{\occV-1}\left( 1-\frac{p}{n}\right)^{U-\occV} \right)\\\nonumber
&=& \brac{1- \frac{p}{n}}^{U-1}\left( \frac{ UE p}{n}\right. \\
& & \;\;\;\;\;\; \left. +\sum_{v\in \mathcal{V}}\left( p\brac{\frac{n}{n-p}}\left( 1-\frac{1}{n}\right)\right)^{\occV}\left(\frac{U p}{n}  + \occV\frac{(1-p+p/n)}{p}\frac{n}{n-1}\left(1-\frac{p}{n}\right) \right)\right) \nonumber  \\
\nonumber
&=&\brac{1- \frac{p}{n}}^{U-1}\left( \frac{UE p}{n}
+\sum_{v\in \mathcal{V}}\left(\frac{p(n-1)}{n-p}\right)^{\occV}\left(\frac{U p}{n}  + \occV\frac{(1-p)}{p}\frac{\left(n-p\right)}{n-1} \right) \right) + O(1).
\\ \label{EYLAZY}
& \le & \brac{1- \frac{p}{n}}^{U-1}\left( \frac{UE p}{n}
+\left(\frac{p(n-1)}{n-p}\right)^{2}\left(\frac{U^2 p}{2n}  + U\frac{(1-p)}{p}\frac{\left(n-p\right)}{n-1} \right)\right)+O(1).
\end{eqnarray}
The last inequality above holds because $p(n-1)/(n-p) \le 1$,
$\occV \ge 2$, $\mathcal{V} \le U/2$ and $\sum_{v\in \mathcal{V}} \occV = U$.

\ignore{
\[(1-p)p^{\occP} \left( 1-\frac{1}{n}\right)^{\occP}\left( 1-\frac{p}{n}\right)^{U-\occP-1}=(1-p)\left(\frac{p(n-1)}{n-p}\right)^{\occP}\left( 1-\frac{p}{n}\right)^{U-1}. \]

Since $\frac{n-1}{n-p}<1$, this is maximised when $\occP=1$ for all unhappy particles, i.e. each unhappy particle shares the vertex it is at with one other particle, and so $\occV=2$ for each such unhappily occupied vertex. The other term in the summation

\[\sum_{v\in \mathcal{V}}\left(\frac{p(n-1)}{n-p}\right)^{\occV}\frac{U p}{n}, \]

is also clearly maximised in this case, and so we have that $\E Y$ is maximised by this case. In particular this gives us that $|\mathcal{V}|\leq U/2$ and,

\[\E Y \leq  \brac{1- \frac{p}{n}}^{U-1}\left( \frac{UE p}{n}
+\left(\frac{p(n-1)}{n-p}\right)^{2}\left(\frac{U^2 p}{2n}  + U\frac{(1-p)}{p}\frac{\left(n-p\right)}{n-1} \right)\right)+O(1).\]
}

As before, the concentration of $H'$ will follows from a martingale argument as each of the $U$ particles chooses whether and where to move independently.
A difference in whether or not to move and choice of vertex by one particle can only alter the final value of $H'$ by at most 2, so
Inequality~\eqref{martingale} applies.


Using~\eqref{EXLAZY} and~\eqref{EYLAZY}, we have the following bound on the expected value of $H'$.
\begin{flalign}\nonumber
\E [ H' \: | \: H]
&\leq  \brac{1- \frac{p}{n}}^{U-1}\left( \frac{EU p+H(n-p)}{n}  +\left(\frac{p(n-1)}{n-p}\right)^{2}\left(\frac{U^2 p}{2n}  + U\frac{(1-p)}{p}\frac{\left(n-p\right)}{n-1} \right)\right)\\\label{LAZY_EDH2}
&\leq  \brac{1- \frac{p}{n}}^{U-1}\left( \frac{EU p+H(n-p)}{n}  +\frac{U^2 p^3}{2n}  + Up(1-p) \right).
\end{flalign}

Let $\delta=\frac{p}{2}-\alpha$, so

\[\NP=\left(1-\delta\right)n=\left(1-\frac{p}{2}+\alpha\right)n.\]

We write the current value of $H$ as

\begin{equation}\label{hjqw}
H= \NP-\varepsilon n=\left(1-\delta-\varepsilon\right)n,
\end{equation}

where $0 < \varepsilon \le 1-\d$. Thus

\begin{equation}\label{hkq2}
U=\varepsilon n \;\; \mbox{ and } \;\; E\leq \left(\delta+\varepsilon\right)n.
\end{equation}

We use \eqref{hjqw} and~\eqref{hkq2} in \eqref{LAZY_EDH2}.

\begin{align}\nonumber
\lefteqn{\E [ H' \: | \: H]} & & & \\\nonumber
&\leq n\brac{1- \frac{(U-1)p}{n}+\frac{U^2p^2}{2n^2}}\left( \frac{EU p+H(n-p)}{n^2}  +\frac{U^2 p^3}{2n^2}  + \frac{Up(1-p)}{n} \right) \\\nonumber
&\leq  n\left(1- \varepsilon p +\frac{p}{n}+\frac{\varepsilon^2p^2}{2}\right)\left(  \left(\delta+\varepsilon\right)\varepsilon p+\left(1-\delta-\varepsilon\right)(1-p/n) +\frac{\varepsilon^2 p^3}{2}  + \varepsilon p(1-p) \right) \\\nonumber
&=  n\left(1+\frac{p}{n}- \varepsilon p\left( 1- \frac{\varepsilon p}{2}\right)\right)\left(\left(1+\delta +\varepsilon \right)  \varepsilon p+1-\delta-\varepsilon +O\left(\frac{p}{n}\right)+\frac{\varepsilon^2 p^3}{2}  - \varepsilon p^2 \right) \\\nonumber
&\leq   n\left(1- \varepsilon p\left( 1- \frac{\varepsilon p}{2}\right)\right)\left( 1-\delta-\varepsilon +\varepsilon p\left( 1+\delta +\varepsilon+ \frac{\varepsilon p^2}{2}-p\right) \right) +O(1) \\\nonumber
&=  n\left(1-\delta -\varepsilon \right. \\\nonumber
& \phantom{= n(} \left. + \varepsilon p\left( 1+\delta +\varepsilon+ \frac{\varepsilon p^2}{2}-p -\left( 1- \frac{\varepsilon p}{2}\right)\left(1-\delta -\varepsilon +\varepsilon p\left( 1+\delta +\varepsilon+ \frac{\varepsilon p^2}{2}-p\right) \right)\right)\right) +O(1)
\\\label{lsio1d}
&=  n\left(1-\delta -\varepsilon + \varepsilon pA\right) +O(1)
\end{align}

We now bound $A$:

\begin{eqnarray}\nonumber
A &=&
1+\delta +\varepsilon+ \frac{\varepsilon p^2}{2}-p -\left( 1- \frac{\varepsilon p}{2}\right)\left(1-\delta -\varepsilon +\varepsilon p\left( 1+\delta +\varepsilon+ \frac{\varepsilon p^2}{2}-p\right) \right)
\\\nonumber
&=&
1+\delta +\varepsilon+ \frac{\varepsilon p^2}{2}-p
- 1+\delta +\varepsilon -\varepsilon p\left( 1+\delta +\varepsilon+ \frac{\varepsilon p^2}{2}-p\right)
\\\nonumber
&&+ \frac{\varepsilon p}{2}
\left(1-\delta -\varepsilon +\varepsilon p\left( 1+\delta +\varepsilon+ \frac{\varepsilon p^2}{2}-p\right) \right)
\\\nonumber
&=&
2\delta +2\varepsilon-p
+\varepsilon p\left( \frac{p}{2} + \frac{1}{2} - \frac{\d}{2} - \frac{\varepsilon}{2}
+ \frac{\varepsilon p}{2}\left( 1+\delta +\varepsilon+ \frac{\varepsilon p^2}{2}-p\right)
- 1-\delta -\varepsilon- \frac{\varepsilon p^2}{2} +p \right)
\\\label{jkkes}
&\le&
2\varepsilon- 2\a
+\varepsilon p\left( \frac{p}{2} - \frac{1}{2} - \frac{3}{2}\d - \frac{3}{2}\e
+ \varepsilon p  - \frac{\varepsilon p^2}{2} +p \right)
\\\nonumber
&\le &
2\varepsilon- 2\a
+\varepsilon p\left( p - \left(\frac{1}{2}-\frac{p}{2}\right) - \frac{3}{2}\d - \e\left(\frac{3}{2} - p\right) - \frac{\varepsilon p^2}{2}\right)
\\\label{eml2q}
&\le &
2\varepsilon- 2\a +2\varepsilon p.
\end{eqnarray}

For Inequality~\eqref{jkkes}, use $\d = (p/2) -\a$ and observe that $1+\d +\varepsilon + (\varepsilon p^2)/2 -p \le 2$.
Combining~\eqref{lsio1d} and~\eqref{eml2q}, we get

\begin{eqnarray}\nonumber
\E [ H' \: | \: H]
&\le & n\left(1-\delta -\varepsilon + \varepsilon pA\right) +O(1)
\\\label{ld90a2}
& \le & n\left(1-\d -\varepsilon + 2\varepsilon p(\varepsilon(1+p) - \a)\right) +O(1).
\end{eqnarray}

This implies that for $\om(1/n) = \e \le \a/(3(1+p)) := \e_0$, we have $2(\varepsilon(1+p) - \a)\leq -4\a/3\leq-\a$ and so,

\begin{equation}\nonumber
\E [ H' \: | \: H] \le n\left(1-\d -\varepsilon - \e p\a\right),
\end{equation}

so from~\eqref{martingale},

\begin{equation}\label{qjks9v}
 \Pr\left( H' \le n(1-\d-\e) \right) \ge 1 -  e^{-\Om(n\e p^2\a^2)}.
\end{equation}

In the range $\e_0 \le \e \le 1 -\d$,
the bound in~\eqref{ld90a2} does not seem strong enough to separate $H'$ from $n(1-\d)$,
but the following bound on $U':= U(t+1)$, which holds for any $0 \le \e \le 1-\d$, will help.

\begin{eqnarray}
\label{jk334l}
\E [U' \: | \: U] & \ge & U\left( 1 - \left(1 - \frac{p}{n}\right)^{U-1}\right)
\\\nonumber
& \ge & U\left( 1 - \left(1 - \frac{(U-1)p}{n} + \frac{U^2p^2}{2n^2}\right)\right)
\\\label{hjkk23n}
& = & U\left( 1 - \left(1 - \frac{Up}{n} + \frac{U^2p^2}{2n^2}\right)\right) +O(1)
\;\; \ge \;\; \frac{U^2p}{3n} \;\; = \;\;  \frac{n\e^2p}{3}.
\end{eqnarray}
Inequality~\eqref{jk334l} holds because the probability that a given unhappy particle $P$
remains unhappy is at least the probability that at least one other unhappy particle decides to move
to the vertex chosen by $P$.

The bound~\eqref{hjkk23n} implies 
\begin{equation}
\E [ H' \: | \: H] \le n\left(1-\d -\frac{\e^2p}{3}\right),
\end{equation}
so
\begin{equation}\label{jkc56a}
 \Pr\left( H' \le n\left(1-\d-\frac{\e^2p}{4}\right) \right) \ge 1 -  e^{-\Om(n\e^3 p^2)}.
\end{equation}

We use bounds~\eqref{qjks9v} and~\eqref{jkc56a} to conclude the proof.
Assume that $\ep \ge \e_0^2p/4$, that is, assume that
$H \le n(1-\d - \e_0^2p/4)$.
Then for some constant $c$,
\[ \Pr \left( H' \le n\left(1-\d - \frac{\e_0^2p}{4}\right) \right) \ge  1 -  e^{- 2 c n p^2\a^3}.
\]
The above inequality follows from~\eqref{qjks9v} for the case when $\e_0^2p/4 \le  \e \le \e_0$,
and from~\eqref{jkc56a} for the case $\e_0 \le \e \le 1-\d$.
Thus the probability that
$H$ reaches $n(1-\d) = M$ and dispersion completes within $e^{c n p^2\a^3}$ steps is at most $e^{- c n p^2\a^3}$.

\section{Dispersion on $k$-regular trees}\label{SkReg}


Let $G$ be a (sufficiently large) $k$-regular complete rooted tree, with the initial position of the particles  at the root.
The root has $k$ children, and each other (internal) vertex has $k-1$ children.
For the process to end, there must be particles that have reached distance at least $\log_{k-1} \NumParticles$,
since there are less than $\NumParticles$ vertices within distance $\log_{k-1} \NumParticles$ from the root.
 We prove the following stronger lower and upper bound, which proves Theorem \ref{kReg}, with

\[\alpha_k=1-\frac{1}{2(\log_{k-1}k) -1} \mbox{ and } \beta_k=\frac{1}{3}-\frac{1}{3\log_{k-1}k} \]

For small values of $k$ there is a gap in the lower and upper bounds, for example,  $\alpha_3\approx 0.54$ and $\beta_3\approx 0.12$, $\alpha_4\approx 0.34$ and $\beta_4\approx 0.07$, but as $k$ increases, these both tend to $0$, making the bounds asymptotically tight.

\begin{theorem}
For all $\varepsilon>0$, on a $k$-regular tree, with $\NumParticles$ particles initially placed at a single vertex,
with probability $1 - O(\NP^{-\ep})$, when the dispersion process terminates, the maximum distance $\DD$ any particle
is from the origin is bounded by

\[ \left(1+\frac{1}{2(\log_{k-1}k) -1} -\varepsilon\right)\log_{k-1} \NumParticles\le \DD \le \left(\frac{5}{3}+\frac{1}{3\log_{k-1}k}+2\varepsilon \right)\log_{k-1}. \NumParticles\]

\end{theorem}

\begin{proof}
As previously mentioned, we predetermine the movements  the particles  take when they move.
An important consequence of this,  is that while a particle may not visit every vertex of its random walk before the dispersion process ends, if the predetermined random walk never visits a particular vertex, then regardless of the behavior of the other particles and the general dispersion process, that particle will never visit that  vertex.

For a particle to reach a distance greater than $d$, it must first reach a vertex at distance $d$. To move on from this vertex, another particle must at some point also be at this vertex. We demonstrate that even if a particle reaches a vertex at depth of $d=(2+\varepsilon)\log_{k-1} \NumParticles$, the probability that any other particle visits this same vertex tends to $0$, hence no particle reaches a depth of $d+1$. This is a weaker upper bound than the one we will prove but is useful for the following observation.

We will make use of a number of results that apply to infinite trees but they can also be made applicable to sufficiently large finite trees. The walks of a particle on the infinite $k$-regular tree and a finite $k$-regular tree (i.e. one that is $k$-regular until terminating in a level of leaves at some depth), are identically distributed until the walk reaches a leaf vertex. Since we will show that in the infinite graph, no two particles ever visit the same vertex at depth $d=(2+\varepsilon)\log_{k-1} \NumParticles$, we can say that in the finite case no two vertices will visit a given vertex at depth $d$ before at least one of them has reached a leaf node. Assuming that the leaf layer is at depth greater than $d$, then this property must still hold in the finite case as no particle can advance beyond $d$ to reach a leaf node. This tells us that any particle will stop walking in the dispersion process before ever reaching a leaf node and so it cannot reach any vertex in the finite case that it would not have reached in the infinite tree. We therefore assume from here onwards that we are working in the infinite tree. Since we will demonstrate a better upper bound on the distance a particle can travel, the same argument shows that these results hold as long as the depth of the leaves is greater than our upper bound on $\DD$.


Consider the predetermined random walk for a single particle. What is the probability that it ever reaches a particular vertex $v$, at depth $d$ from the root?
Let $X_t$ be the current distance of the particle from $v$ at time $t$, so $X_0=d$. With probability $p=\frac{k-1}{k}$, the particle will move away from $v$ and as such $X_t$ will increase, conversely, with probability $q=1-p=\frac{1}{k}$, the particle will move towards $v$, decreasing $X_t$. The properties of random walks with bias are given in
Feller Chapter XIV, \cite{Feller_Vol1_2Edition}. By equation (3.6) of that chapter, the probability of reaching $v$ (ultimate ruin) starting from distance $d$ is

\begin{equation}\label{Eqn:Probabilitybound}
	\Pr(X_t=0, \mbox{ for some } t)
=\left(\frac{q}{p}\right)^d=\left(\frac{1}{k-1}\right)^d.
\end{equation}


Consider a single particle and let $v$ be the first (if any) vertex it visits at depth  $d=(2+\varepsilon)\log_{k-1} \NumParticles$ during the dispersal process. Using, \eqref{Eqn:Probabilitybound}, the probability that another particle reaches $v$ satisfies

\[\left(\frac{1}{k-1}\right)^d=\left(\frac{1}{k-1}\right)^{(2+\varepsilon)\log_{k-1} \NumParticles}=\frac{1}{\NumParticles^{2+\varepsilon}}. \]

Taking a union bound, we see that the probability that a particle reaches depth $d$ for the first time at a position that another vertex may visit at any point in its random walk is less than $\binom{\NumParticles}{2} \frac{1}{\NumParticles^{2+\varepsilon}}<\NumParticles^{-\varepsilon}\rightarrow 0$ as $\NumParticles\rightarrow \infty$ as required.

For the lower bound, we require the following lemma.

\begin{lemma}\label{Lemma:No3Particles}
For all $\varepsilon>0$, and a given $d\geq \frac{3+\varepsilon}{2} \log_{k-1} \NumParticles$,
with probability $1-O(\NP^{-\ep})$ no vertex at depth $d$ is ever visited by more than $2$ distinct particles.
\end{lemma}
\begin{proof} We again make use of \eqref{Eqn:Probabilitybound} and the fact that there are $k(k-1)^{d-1}$ vertices at depth exactly $d$, to observe that the probability that any three particles ever visit a common (but unspecified) vertex at depth $d$ is less than

\[\binom{\NumParticles}{3} \left(\frac{1}{k-1}\right)^{3d}  k(k-1)^{d-1}\leq \frac{k}{k-1}\NumParticles^3 (k-1)^{-2d}.\]

For $d\geq\frac{3+\varepsilon}{2} \log_{k-1} \NumParticles$,
the right-hand side of the above inequality is $= O(\NP^{-\ep})$.
\end{proof}

The number of vertices at depth $d$ or less is equal to $(k(k-1)^d-2)/(k-2)$.
For $d=(1-\varepsilon/2)\log_{k-1} \NumParticles$, this is equal to $\frac{k\NumParticles^{1-\varepsilon/2}-2}{k-2}=o(\NumParticles)$. Therefore at least $\NumParticles-o(\NumParticles)$ particles must at some point before the process ends, reach depth greater than $(1-\varepsilon/2)\log_{k-1} \NumParticles$.
Each such particle must at some time, be at depth $d$ with at least one other particle at the same vertex.

We will pair up these particles and consider the probability that they continue together to a further depth. The first time two particles are at depth $d$ at the same vertex, we take them as the first pair (choosing arbitrarily if there are multiple choices).

Given this configuration, we calculate the probability that these two particles now advance a further $d'-1=(A\log_{k-1} \NumParticles)-1$ steps together without separating and then separate on the $d'$th step.  The probability of this event is equal to

\begin{align}
\left(\frac{k-1}{k}\times\frac{1}{k}\right)^{d'-1}\left( \frac{k-1}{k}\right)=&k\left(\frac{k-1}{k^2}\right)^{(A\log_{k-1} \NumParticles)}\nonumber\\
=& k\left(k-1\right)^{(1-2\log_{k-1}k)A\log_{k-1} \NumParticles}\nonumber\\
=&k\NumParticles^{-((2\log_{k-1}k)-1)A}.\label{Eqn:ProbOfAdvancing}
\end{align}

If this event occurs, then both particles reach a vertex at depth $d+d'-1$ and at least one advances to a depth of $d+d'$ before they separate. We can determine whether this event happens by looking only at the pre-determined walks for these particles until they successfully advance and separate or the first time at which they fail to do so. This ensures any future behaviour of the pre-determined walk of the particles is still random and the above is independent of the behaviour of any other particle.

If the event does not occur, then we consider the next pair of particles. We choose them to be the two particles that will next be at a common vertex at depth $d$. This may occur at the same time as the previous pair. As we only examine the predetermined walks for the pair of particles each time, the probability of these two particles advancing is independent of and hence equal to that calculated in \eqref{Eqn:ProbOfAdvancing}. If there were three particles at this vertex, then we simply ignore the third particle and move on to the next pair (equally for any odd number, we pair the particles up and discard the final particle).

If the next pair reaches $d$ at a later time, it is possible that one (or both) of the vertices had already been included in a previous pair that failed to advance and returned to depth $d$. The probability that this pair advances is still independent, as once a particle had returned to depth $d$, it must have failed to advance and as such we stopped examining the next steps in its predetermined walk and have preserved the randomness of the particles movement from this point on.

Each pair that we consider, therefore uses up at most $3$ of the $\NumParticles-o(\NumParticles)$ particles that must at some time be at depth $d$ at the same time as another particle, and so we have at least $\NumParticles/3-o(\NumParticles)\geq \NumParticles/4$ such pairs to consider, where the probability that each pair advances before separating is independent.
Using \eqref{Eqn:ProbOfAdvancing}, setting $A=\frac{1-\varepsilon/2}{2(\log_{k-1}k) -1}$ we have the probability that none of these pairs advance is less than,

\[\left(1-k\NumParticles^{-((2\log_{k-1}k)-1)A}\right)^{\NumParticles/4}\leq e^{-\frac{k\NumParticles^{1-((2\log_{k-1}k)-1)A}}{4} }=e^{-\frac{k\NumParticles^{\varepsilon/2}}{4} }. \]

Therefore, with high probability, at least one of these pairs advances to a total depth of
\[d+d'=(1-\varepsilon/2+\frac{1-\varepsilon/2}{2(\log_{k-1}k) -1})\log_{k-1} \NumParticles\geq (1-\varepsilon+\frac{1}{2(\log_{k-1}k) -1})\log_{k-1} \NumParticles,\]

(since $2(\log_{k-1}k)-1>1$), and then separates, leaving the two particles on separate vertices. Since they are on different vertices, for either of the particles to move again, a third particle would have to visit the location of that particle. This implies that three distinct particles would have to visit the vertex prior to the one on which the pair separated,
and by Lemma \ref{Lemma:No3Particles}, this can occur only with probability $O(\NP^{-\ep})$.

Therefore, with probability $1 - O(\NP^{-\ep})$, at least one vertex ends the process at depth
\[(1-\varepsilon+\frac{1}{2(\log_{k-1}k) -1})\log_{k-1} \NumParticles.\]

We now return to prove the stronger corresponding upper bound. Consider vertices at depth $d\geq(\frac{3}{2}+\frac{1}{6}+\varepsilon) \log_{k-1} \NumParticles=(\frac{5}{3}+\varepsilon) \log_{k-1} \NumParticles$. No three particles have any vertex at this depth in all three of their pre-determined paths. With probability $1$ all particles will eventually reach this depth in their infinite walk and so each particle has a corresponding vertex that it first reaches at this depth at some point in their pre-determined walk. Given that at most two particles can reach such a vertex, for a particle to reach this depth in the dispersion process and continue moving, exactly one other particle must also have this vertex in its path. If this vertex is not the first at depth $d$ in the second particle's path, then to reach this point it would have to have reached some other vertex at $d$ and then returned, requiring another particle to have been at this second vertex and for both particles to move together without separating, at least until they returned to depth $\frac{3+\varepsilon}{2} \log_{k-1} \NumParticles$ where they might encounter another particle. If we can demonstrate that no pairs of particles can return from this depth, then we only need consider particles that share a vertex as their first visit to this depth.

Since all particles reach exactly one vertex at this depth, independently over both vertices and particles, and uniformly  at random, the number of vertices where two particles first reach depth $d$ in their walks is distributed exactly as the number of collisions in a balls in bins problem, with $\NumParticles$ balls going into $k(k-1)^d$ bins. (Noting that all collisions are of exactly $2$ particles, since no three have common vertices in their paths at this depth.)

The total number of vertices at this depth is $\Theta \left( \NumParticles^{\frac{5}{3}+\varepsilon}\right)$. Examining each particle one by one, we see there at most $\NumParticles$ vertices already claimed by a particle and so the probability that it produces a collision is at most $\Theta \left(\NumParticles/\NumParticles^{\frac{5}{3}+\varepsilon}\right)=\Theta \left(\NumParticles^{-\left(\frac{2}{3}+\varepsilon\right)}\right)$, and this bound holds independently of the outcome of the other trials. Therefore we can bound the number of collisions by a sum of independent random indicator variables each with expectation $\Theta\left(\NumParticles^{-\left(\frac{2}{3}+\varepsilon\right)}\right)$. Let $X$ be the value of this sum of $\NumParticles$ variables, and so we have

\[\E(X)=\Theta\left(\NumParticles\times \NumParticles^{-\left(\frac{2}{3}+\varepsilon\right)}\right)=\Theta\left( \NumParticles^{\frac{1}{3}-\varepsilon}\right).\]
By Hoeffding's inequality,

\[\Pr\left[X>(1+\varepsilon)E(X) \right]\leq e^{-2\NumParticles \varepsilon^2 E(X)^2}. \]

Thus with high probability we have that the expected number of collisions is of the order $\Theta\left(\NumParticles^{\frac{1}{3}-\varepsilon}\right)$.

Even assuming that each of these collisions occurs during dispersion (i.e. both particles reach each of these vertices in the process), for a given pair of two particles to return from depth $d$ to $\frac{3+\varepsilon}{2} \log_{k-1} \NumParticles$, they would have to move together a distance upwards of at least $\frac{1}{6}\log_{k-1} \NumParticles$ without separating. In particular both particles would have to contain a common ancestor vertex at some later time in their pre-determined paths at distance $\frac{1}{6}\log_{k-1} \NumParticles$ above the vertex at which they meet. The movements of the particles after the point at which they first meet at this depth are independent random walks and so by \eqref{Eqn:Probabilitybound}, the probability that one of the particles visits the ancestor vertex at any point in it's infinite walk after this point is equal to

\[(k-1)^{-\frac{1}{6}\log_{k-1} \NumParticles}=\NumParticles^{-\frac{1}{6}}.\]

Therefore, the probability that both particles visit this ancestor is $\NumParticles^{-\frac{1}{3}}$, and by taking a union bound over the $\NumParticles^{-\frac{1}{3}-\varepsilon}$ pairs, we see that with probability $1 - O(\NP^{-\ep})$
no such pair can return to this depth.

We claim that no such pair of particles will move more than

\[\left(\frac{1+\varepsilon}{3\log_{k-1} k}\right)\log_{k-1} \NumParticles=\left(\frac{1}{3}+\varepsilon\right)\log_{k} \NumParticles,\]

steps before separating. Note this is larger than the distance required to return to the earlier depth, but with high probability most of these steps will be in the wrong direction, advancing further down the tree. The probability that two particle move distance $c \log_{k-1} \NumParticles$ without separating is equal to $k^{-c \log_{k-1} \NumParticles}=\NumParticles^{-c\log_{k-1} k}$. Taking a union bound we see that the probability that any of the pairs advance further than the above is less than

\[\NumParticles^{\frac{1}{3}-\varepsilon}\NumParticles^{-\left(\frac{1}{3}+\varepsilon\right)}=\NumParticles^{-2\varepsilon}.\]

Since the particles separate before returning to depth $\frac{3+\varepsilon}{2} \log_{k-1} \NumParticles$, then their positions do not intersect with the paths of any other particles, and so they will not move any further in the process once they have separated.

Lastly, this tells us that any pair of particles that manages to reach a depth of $(\frac{5}{3}+\varepsilon) \log_{k-1} \NumParticles$ will then move at most $\left(\frac{1}{3}+\varepsilon\right)\log_{k} \NumParticles$ further steps.
Therefore, at the end of the process, with probability $1 - O(\NP^{-\ep})$ no particle will have reached a depth higher than

\[\left(\frac{5}{3}+\varepsilon\right) \log_{k-1} \NumParticles+\left(\frac{1}{3}+\varepsilon\right)\log_{k} \NumParticles \le \left(\frac{5}{3}+\frac{1}{3\log_{k-1}k}+2\varepsilon \right)\log_{k-1} \NumParticles.\]

%

\end{proof}

\section{Dispersion on paths} \label{SPath}

In this section we prove Theorem~\ref{Path}.
To analyse the case of a sufficiently large or infinite path,
we consider that $G$ is the integer line, i.e $V(G)=\mathbb{Z}$ and $(i,j)\in E(G)$ if and only if $|i-j|=1$. All $\NumParticles$ particles are initially placed at the origin $0$.
We will require the following result about random walks on the line.

\begin{lemma}\label{QVal}
For a simple random walk on the integer path, let $R(2T,r)$ be the probability of at least $r$ returns to the origin in $2T$ steps. For $\alpha>0$, we have
\begin{equation}\label{prob}
R(2T,\a\sqrt{2T}) \le
O\brac{\frac{1}{\a} e^{-\a^2/2}}.
\end{equation}
\end{lemma}
\begin{proof}

Let $z_{2T}^{(r)}$ be the probability of exactly $r$ returns to zero in $2T$ steps, then for $T \ge 1$, we have from Feller, Theorem 1, Section 3.6 in \cite{Feller_Vol1_2Edition}
\[
F_r= z_{2T}^{(r)}=\frac{1}{2^{2T-r}}{2T-r \choose T}.
\]
We need to calculate
\[
R(2T,r)=\sum_{s \ge r}F_s,
\]
namely, the probability of at least $r$ returns in time $2T$, for the case when $r=\a\sqrt{2T}$. If $r=T$ then $F_T=1/2^T$ (very small).

For $s \ge r\ge 1$
\[
F_{s+1}
/F_s= \frac{2T-2s}{2T-s} \le \frac{2T-2r}{2T-r}=F_{r+1}/F_r=\b.
\]
So
\[
R(2T,r) \le F_r(1+\b+\b^2+...)=F_r \frac{1}{1-\b}=F_r \frac{2T-r}{r}.
\]
Now, provided $r<T$
\begin{flalign*}
F_r=& (1+(1/T)+(1/(T-r)) \frac{1}{\sqrt{2\pi}}\sqrt{\frac{2T-r}{T(T-r)}}
\frac{2^r}{2^{2T}}\frac{(2T-r)^{2T-r}}{T^T(T-r)^{T-r}}
\\&= \Th(1)\sqrt{\frac{2T-r}{T(T-r)}}
\frac{(1-r/2T)^{2T-r}}{(1-r/T)^{T-r}}.
\end{flalign*}
But\begin{flalign*}
\frac{(1-r/2T)^{2T-r}}{(1-r/T)^{T-r}}=&\exp\brac{(2T-r) \log(1-r/2T)-(T-r)\log(1-r/T)}
\\
&=\exp\brac{-\frac{r^2}{4T}-\frac{r^3}{8T^2}-\cdots}\\
&\le \exp\brac{-\frac{r^2}{4T}}.
\end{flalign*}
Thus
\[
R(2T,r) \le \Th(1)\sqrt{\frac{2T-r}{T(T-r)}}\frac{2T-r}{r}\exp\brac{-\frac{r^2}{4T}}.
\]
Put $r=\a\sqrt{2T}$ to obtain (assuming $r<T$)
\[
R(2T,\a\sqrt{2T}) =O\brac{\frac{1}{\a} e^{-\a^2/2}},
\]
as required.

\end{proof}

The movement  of particle $i$ at time step $t$ takes a value $X_i(t)\in \{-1,0,1\}$, with $X_i(t)=0$ if the particle doesn't move.
For any particle  the next non-zero movement  is uniformly  distributed in $\{ -1,1\}$ and independent of the choice of particle, or of the action of any other particles.

If we consider only the steps where a given particle moves (walk steps), and ignore the time-steps in which the particle does not move, the particle makes a random walk on the line.
The particle moves only when its random walk intersects that of another particle. When two particles meet at a vertex,  reversing the walk of the second particle and taking the union of these two walks, gives a walk which has returned to the origin.

We consider the walk steps of  particles $1$ and $2$, and  build a sequence $Y(t)=Y(t,\{1,2\})$ as follows. Let $W \in \{-1,1\}$ denote the movement of a particle at a given walk step.
Note that if $W$ is uniformly and independently distributed on $\{-1,1\}$, then so is $-W$.
The entries of
$Y(t)$ are the movements $(W_1),(W_2)$ made by the two particles up to the end of time step $t$
in time step order. If both particles move at a given time step, the order is the movement of particle 1 followed by that of particle 2. The entries for particle 2 are the negative of the step direction $W_2$. Thus
$Y(0)=(W_1(0), -W_2(0))$ as both move from the origin. In general, $Y(t)=Y(t-1)$ if neither move, $Y(t)=(Y(t-1),W_1(t))$ if particle 1 moves but not particle 2,
$Y(t)=(Y(t-1),-W_2(t))$ if particle 2 moves but not particle 1, $Y(t)=(Y(t-1),W_1(t), -W_2(t))$ if both move.

Let $s_i(t)$ be the number of walk steps taken by  particle $i$ at the end of time step $t$. The length of $Y(t)$ is $s=s_1(t)+s_2(t)$. As $t \rai$, either $Y$ is infinite (dispersion never stops) or has a finite length $T$,  and $Y=(Y_1, Y_2,...,Y_T)$.
If $T$ is finite, extend $Y$  for $i>T$ by setting $Y_i\in\{-1,+1\}$,  chosen independently with probability $1/2$.

Note that without knowledge of the value of $T$,  each $Y_i$ is distributed uniformly and independently at random and so $Y$ is an infinite random walk on the line. Importantly, $Y^s=\sum_{i=1}^s Y_i$ visits the origin at least once for each time step when the original two particles intersected
up to the time step when they have moved $s$ walk steps in total.
Note that $Y$ may visit the origin more often than the particles intersect as they move simultaneously, so $Y$ can hit the origin for values of $i$ that lie between two simultaneous movements of the particles in one time step of the dispersion process.
\ignore{
\begin{theorem}
For all $\varepsilon>0$, given $\NP$ particles initially placed at the origin,
w.h.p. the dispersion process on the integers will terminate with no particle  at a distance greater than $(4+\varepsilon)\NP\log \NP$ from the origin.
\end{theorem}}
\begin{theorem}
For all $\varepsilon>0$, w.h.p. the dispersion process on the integers with $\NP$ particles at $0$, will terminate with no particle at a distance greater than $(4+\varepsilon)\NP\log \NP$ from the origin.
\end{theorem}
\begin{proof}
At each time step $t \ge 0$ let $Z_i(t)$ be the walk length, i.e. the total number of  walk steps made by  particle $i$ up to the end of time step $t$.
 Let $K$ be a large constant, and let $F_{ij}$ be the event that
\[
F_{ij}=\{\text{Particles }i,j \text{ with walk lengths } Z_i,Z_j \le S+K \text{
meet more than }\alpha (2S)^{1/2} \text{ times}\},
\]
and let $F=\cup_{i \ne j} F_{ij}$.
Using the sequence $Y(t, \{i,j\})$ for the particles $i,j$ (as described above for particles 1,2) we can upper bound the number of  meetings between the particles, by the number of returns to the origin of the  random walk $Y(t)$.
As returns are monotone non-decreasing with the number of walk steps,
if $Z_i+Z_j < 2(S+K)$ we can extend $Y$ to $2(S+K)$ and include any extra returns.

Let $\alpha^2=4(1+\epsilon) \log \NP$, and $S=2\alpha^2\NP^2$. Thus
 $\alpha S= ((1-O(K/S))\alpha (S+K)$. Using \eqref{prob} of Lemma \ref{QVal} with
$T=S+K$, $\beta=\alpha (1-O(K/S))$ gives that the probability of at least $\beta\sqrt{2(S+K)}$ returns in $2(S+K)$ steps, satisfies,
\[
R\left(2(S+K),\beta\sqrt{2(S+K)}\right)=O\left(\NP^{-(2+\epsilon)}\right).
\]
Thus
\[
\Pr(F) \le \NP^2O(\NP^{-(2+\epsilon)})=O(\NP^{-\epsilon}).
\]
Suppose there exits a particle which takes more than $S$ walk steps, then we pick the first particle (in process time steps $t$) to make $S+1$ walk steps, choosing the particle with the lowest label $I$ if there is any choice.

Recall that $Z_i(t)$
is the number of
 walk steps made by  particle $i$ by time step $t$, and let $R_{ij}(t)$ be the number of
 walk steps when particles $i$ and $j$ occupy the same vertex. Note that each such pair-wise meeting causes both particles to make one step of a random walk, so every walk step is counted at least once. Thus $Z_i(t) \le \sum_{j \ne i} R_{ij}(t)$.
We assume the event $F^c$ holds.
At $t$ for all $i$
\[
Z_i \le \sum_{j \ne i} R_{ij}\le \NP \alpha (2S)^{1/2}.
\]
In particular
\[
Z_I =S+1 \le \NP\alpha (2S)^{1/2},
\]
and thus
\[
(S+1)^2 \le 2 \alpha^2 \NP^2 S=S^2.
\]
This implies that $(S+1) \le S$, contradicting the existence of a first $t$ where some particle exceeds $S$ steps and therefore no particle takes more than $S$ walk steps during the process.

By using a Chernoff bound for the sum of $s$ independent and uniform $\{-1,1\}$ random variables, we see that the probability a random walk reaches a distance greater than $a$ in $s$ walk steps is less than $2e^{-\frac{a^2}{2s}}$.
Suppose a particle is at a distance greater than $4(1+\varepsilon)\NP\log \NP$ from the origin. We must have that the particle either took more than $8(1+\varepsilon) \NP^2\log \NP$ steps, or otherwise, by the above, we would have the probability of this occurring to be less than
\[2e^{-\frac{\left(4(1+\varepsilon)\NP\log \NP\right)^2}{16(1+\varepsilon) \NP^2\log \NP}}=
2e^{-(1+\varepsilon)\log \NP}=2\NP^{-(1+\varepsilon)}.\]
Taking the union bound, with high probability no particle could have reached a distance of $4(1+\varepsilon)\NP\log \NP$.

This bound applies to the maximum distance any particle will be at from the origin at the end of the process. It is possible that a particle may reach a further distance and return before the process terminates. Taking a union bound over the steps of the particle's walk, we can use the same argument to show that at no point in the process, could any particle reach
a distance of $4\sqrt{2}(1+\varepsilon)\NP\log \NP$. Therefore this process will disperse in the same manner on the infinite line as on any finite path or cycle of size larger than $8\sqrt{2}(1+\varepsilon)\NP\log \NP$, which is less than $12\NP\log \NP$ for $\ep = 0.2$.
 \end{proof}

 \section{Dispersion on grids and infinite Cayley graphs}\label{Grid}

In this section we prove Theorem~\ref{Caley}, part $(i)$ and show its implication for the dispersion in
the $2$-dimensional grid.

\begin{lemma}
Let $\om = \om(\NP) \rai$.
Let $G$ be a $d$-dimensional (infinite) grid ($d \ge 1$) or other infinite Cayley graph,
and let $t$ be such that $t \ge \om {\NP \choose 2} R(2t)$,
where $R(2t)$ is the expected number of returns to the origin in $2t$ steps by a simple random walk on $G$.
Then with probability at least $1 - 1/\om$,
a system of $\NP$ particles disperses on $G$ in $t$ process steps
\end{lemma}

\begin{proof}
We use the same argument as for the line, linking times that two particles meet
in a grid or a Cayley graph
with the number of returns to the origin of a single combined random walk.

For ease of comprehension, we work with the $2$-dimensional grid, making comments to show that the arguments apply to a Cayley graph.
The movement of particle $i$ at time step $t$ takes a value $X_i(t)\in \{(-1,0),(1,0),(0,0),(0,-1), (0,1)\}$, with $X_i(t)=(0,0)$ if the particle doesn't move.
For any particle  the next non-zero movement  is uniformly  distributed in $\{(-1,0),(1,0),(0,-1), (0,1)\}$ and independent of the choice of particle, or of the action of any other particles.
In a Cayley graph, $X_i(t) \in S$, where $S$ is the symmetric set of generators which define the graph.

If we consider only the steps where a given particle moves (walk steps), and ignore the time-steps in which the particle does not move, the particle makes a random walk on the grid (or a Cayley graph).
The particle moves only when its random walk intersects that of another particle. When two particles meet at a vertex,  reversing the walk of the second particle and taking the union of these two walks, gives a walk which has returned to the origin.
For the case of a Cayley graph, we need at this point the assumption that the underlying group is abelian.

We consider the walk steps of  particles $1$ and $2$, and  build a sequence $Y(t)=Y(t,\{1,2\})$ as follows. Let $W \in \{(-1,0),(1,0),(0,-1), (0,1)\}$ denote the movement of a particle at a given walk step.
Note that if $W$ is uniformly and independently distributed on $\{(-1,0),(1,0),(0,-1), (0,1)\}$, then so is $-W$.
The entries of
$Y(t)$ are the movements $(W_1),(W_2)$ made by the two particles up to the end of time step $t$
in time step order. If both particles move at a given time step, the order is the movement of particle 1 followed by that of particle 2. The entries for particle 2 are the negative of the step direction $W_2$. Thus
$Y(0)=(W_1(0), -W_2(0))$ as both move from the origin. In general, $Y(t)=Y(t-1)$ if neither move, $Y(t)=(Y(t-1),W_1(t))$ if particle 1 moves but not particle 2,
$Y(t)=(Y(t-1),-W_2(t))$ if particle 2 moves but not particle 1, $Y(t)=(Y(t-1),W_1(t), -W_2(t))$ if both move.

Let $s_i(t)$ be the number of walk steps taken by  particle $i$ at the end of time step $t$. The length of $Y(t)$ is $s=s_1(t)+s_2(t)$. As $t \rai$, either $Y$ is infinite (dispersion never stops) or has a finite length $T$,  and $Y=(Y_1, Y_2,...,Y_T)$.
If $T$ is finite, extend $Y$  for $i>T$ by setting $Y_i\in\{(-1,0),(1,0),(0,-1), (0,1)\}$, each chosen independently with probability $1/4$.

Note that without knowledge of the value of $T$,  each $Y_i$ is distributed uniformly and independently at random and so $Y$ is an infinite random walk on the grid. Importantly, $Y^s=\sum_{i=1}^s Y_i$ visits the origin at least once for each time step when the original two particles intersected
up to the time step when they have moved $s$ walk steps in total.
Note that $Y$ may visit the origin more often than the particles intersect as they move simultaneously, so $Y$ can hit the origin for values of $i$ that lie between two simultaneous movements of the particles in one time step of the dispersion process.

\ignore{
Considering this, we can see that the probability that the two particles meet more than $k$ times is less than the probability that the combined walk visits the origin more than $k$ times.
}

Let $r_D(s)$ be the probability $Y^s=(0,0)$ and let $r(s)$ be the probability of a return to the origin at step $s$ of a simple random walk on the grid. Then
\[
\sum_{s=0}^t r_D(s) \le \sum_{s=0}^{2t} r(s) = R(2t),
\]
where $R(2t)$ is the expected number of returns to the
origin of a random walk during $2t$ steps.

For a system of $\NP$ particles dispersing from the origin,
the above discussion bounds the expected number of meetings of a given pair of
particles in $t$ steps by $R(2t)$.
Let $Z(t)$ be the number of pairwise meetings between $\NP$ particles
in $t$ process steps. Then,
\[
E(Z(t)) \le {\NP \choose 2} R(2t)
\]
Given $\omega \rai$,
take any $t$ satisfying
\begin{equation}\label{hcwiqbc3}
 {\NP \choose 2}R(2t) \le  t/\omega.
\end{equation}
By Markov's inequality we have
$P(Z(t) \ge \omega \E(Z(t)) \le 1/\omega$.
Then  with probability $1-O(1/\omega)$
\[
Z(t) \le \om \E(Z(t)) \le \omega {\NP \choose 2} R(2t)< t.
\]
Since at least one pair of particles moves during each step of dispersion,
the number of meetings $Z(t')$ must be larger than $t'$ for each process step $t'$,
therefore the process must have stopped before step $t$.
\end{proof}

\begin{lemma}
With probability $1- O(1/\om)$, a system of $\NP$ particles disperses on the 2-dimensional grid in $2 \om \NP^2 \log \NP$
process steps.
\end{lemma}

\begin{proof}
The value of $R(2t)$ is (see e.g. \cite{Feller_Vol1_2Edition} page 328)
\[
R(2t)=\sum_{s=0}^t {2s \choose s}^2 \frac{1}{4^{2s}} \le \log t +c,
\]
for some  constant $c$.
By inspection we see that $t = 2 \om \NP^2 \log \NP$ satisfies~\eqref{hcwiqbc3}.
\end{proof}

\section{Hypercube and finite Cayley graphs}\label{Sec:Hypercube}

In this section we present two different bounds for dispersion on the hypercube using two different methods.

The first method is generally applicable with little modification to general finite Cayley graphs. The method follows a similar structure to that used in the previous section for infinite graphs. The key difference is the difficulty in bounding the number of meetings of particles in a finite graph. To work around this, we first allow the dispersion process to approach mixing time until the probability of being at any given vertex is close to uniform, assuming a trivial bound on the number of meetings in this period. Once we have reached this mixing time, we can use this uniformity to derive a bound on two particles being at the same vertex.

The proof of the lemma below can be easily generalised to a proof of Theorem~\ref{Caley} part $(ii)$.

\begin{lemma}
Let $H_d$ be the hypercube on $n=2^d$ vertices.
Then with probability $1 - O(1/\om)$, a system of $\NP \le \sqrt{n}/\om$ particles disperses on the hypercube $H_d$ in
$O(\NP^2\log^2 n)$ process steps.
\end{lemma}

\begin{proof}

The hypercube $H_d$ on $n=2^d$ vertices consists of vertices labeled as vectors in $\{0,1\}^d$ and edges $uv$ between vertices $u$ and $v$ whenever the vertex labels differ in a single coordinate (Hamming distance one).
Let $e_j$ be the vector whose entries are zero except at the $j$-th coordinate whose entry is one. A transition $(u,v)$ of a random walk on $H_d$ can be modeled by sampling $j \in \{1,...,d\}$ uar and setting $v=(u+e_j)$mod $2$. We refer to $e_j$ as the transition vector.

Mimicking the argument for the grid given previously, we consider the walk steps of  particles $1$ and $2$, and  build a sequence $Y(t)=Y(t,\{1,2\})$ consisting of the transition vectors of the walks of particle 1 and particle 2 at each step in that order. Because $e_j+e_j=0$ mod $2$ we do not need to multiply the second transition vector by $(-1)$ as in the case of the grid.

Let $T$ be a step such that for all even $s \ge T$, then $|P^s(u,u)-1/n| \le 1/2n$ for all $u \in V$. Then $T=O(\log^2 n)$. Let $\NP\le \sqrt{n}/\om$ and $t = T\rdup{\NP^2/(1-\NP^2c\om/2n)}$, where $c \le 3/2$. The expected number of meetings between a pair of particles in $T+t$ steps is at most the expected number of returns to the origin of the walk $Y(t)$ which is at most $2T+2ct/n$.
As in the argument above, with probability
$1-O(1/\om)$ the total number of pairwise meetings in $T+t$ steps is at most
\[
{\NP \choose 2} 2T + \om {\NP \choose 2} c \frac{2t}{n} < T+t,
\]
which is true provided e.g.
\[
2 T\NP^2 \le  t.
\]

\end{proof}

We can improve (for most values of $\NP$ and $n$) on this upper bound for dispersion on the hypercube by using the technique we used in analysing regular trees.

\begin{lemma}
Let $H_d$ be the hypercube on $n=2^d$ vertices.
Then with probability $1 - O(1/\om)$, a system of $\NP \le \sqrt{n}/\om$ particles disperses on the hypercube $H_d$ in
$O(\NP\log^3 n  )$ process steps, and no particle will lie further than $d/2$ from the origin.
\end{lemma}

\begin{proof}

Note that if some particle reaches a distance greater than $c$ from the origin, it must have first reached distance $c$ and there met or later be joined by another particle for it to have moved on. If this second particle had previously already reached distance $c$ before visiting the common vertex where these particles meet, then it must have met some other particle at the place where it first reached distance $c$. Iterating, we see that at some point, there must exist a pair of particles who first reached distance $c$ at the same vertex.

There are $\binom{d}{c}$ vertices at distance $c$ from the origin, and the event that a given vertex is the first a particle visits at that distance (in it's predetermined walk) is uniformly distributed. Therefore the probability of there existing a pair of particles sharing a common vertex as their first visit to distance $c$, tends to $0$ if the number of particles $\NP\ll \sqrt{\binom{d}{c}}$.

Let $c=\frac{d}{2}$. Therefore, we have $\binom{d}{c}=\binom{d}{d/2}\geq \frac{2^{d-1}}{\sqrt{d/2}}$. Since $n=2^d$ we have $\NP \le \sqrt{n}/\om=2^{d/2}/\om$. The probability that two particles share a common vertex as their first visit to distance $c$, is equal to $\binom{d}{c}^{-1}\leq \frac{\sqrt{d/2}}{2^{d-1}}$ and so by taking a union bound over all possible pairs, we have the probability of there existing a pair of particles sharing a common vertex as their first visit to distance $c$ is at most

\[\binom{\NP}{2}\frac{\sqrt{d/2}}{2^{d-1}}\leq \frac{\NP^2 \sqrt{d/2}}{2^{d-1}}\leq \frac{\sqrt{2d}}{\om^2}. \]

Given $\om\gg \sqrt[4]{d}$, this tends to $0$ as required and as such with high probability, no particles reach a distance greater than $d/2$ from the origin in the dispersion process. In fact we are able to say something stronger. This results says that for every particle, it will not move further than the first visit to $d/2$ in its predetermined random walk. This allows us to give a bound on the running time of the process.

Consider the walk of a given particle. We note that if a particle is at distance less than $d/2$ from the origin, then since each vertex has degree $d$, it will move away from the origin with probability at least $1/2$, increasing its distance. Therefore the probability that this walk reaches distance $d/2$ in some number of steps is at least as high as that for a simple random walk on the line. The expected time it takes the simple walk on the line to reach distance $d/2$ is $(d/2)^2$ (by a simple martingale optional stopping time argument) and so the probability that either walk takes more than $\om(d/2)^2$ steps to reach a distance of $d/2$ is less than $\om^{-1}$. In particular, the probability that it takes more than $d^2/2$ steps is less than $1/2$ and so with probability at least $1/2$ the walk will have reached $d/2$ by time $d^2/2$. If not, then we restart the analysis, treating the walk from this time onwards as a new random walk. Although we may not be at the origin, this only reduces the expected time to reach distance $d/2$ and so again with probability at least $1/2$, independently of the previous round, we will reach distance $d/2$ in the next $d^2/2$ steps. We iterate this process $d/2$ times, taking at most $d^3/4$ steps in total, and so the probability the walk has not finished after all these rounds, is at most,

\[(1-1/2)^{d/2}=2^{-d/2}=n^{-1/2}\leq (\om \NP)^{-1} .\]

Taking a union bound over the $\NP$ particles, we see that with probability $1 - \om ^{-1}$, each particle will have reached distance $d/2$ and hence stopped moving in the process after taking at most $d^3/4$ steps.

Since at any time step, at least two particles must be moving or the process has ended, we have that after $(\NP/2)(d^3/4)$ time steps, each particle must have moved at least $d^3/4$ steps and hence the process will end. This tells us that with probability at least
$1 - \om ^{-1}$, the entire process will terminate in $O(\NP\log^3 n)$
 process steps as required.

 \end{proof}

\section{Acknowledgements}
We would like to thank Tony Johansson and Fiona Skerman for spotting a mistake in the manuscript of the paper and for offering suggestions on how to correct it.

\ignore{

\section{Concluding remarks}\label{Sim}

The question of the largest proportion of particles which can be dispersed
in polynomial time on finite graphs seems interesting.
For both the complete graph and the star, the dispersion number is $1/2$.
However the experimental evidence for paths and cycles suggests a dispersion number of at least $0.89$.
Figure \ref{plots} shows 20000 particles  dispersed more or less uniformly over 22500 vertices. The average number of steps made is higher for those particles which finish nearer the origin.

\begin{figure}[H]
\centerline{
\includegraphics[height=140pt, width=180pt]{figure2-20000.pdf}
\hspace*{1cm}
\includegraphics[height=140pt, , width=180pt]{figure1-20000.pdf}
}
\caption{
Simulation results for the Path graph. Left hand figure. Cumulate density at dispersion. Right hand figure. Walk steps taken against final position.}\label{plots}
\end{figure}
}


\end{document}